%% file: Main.tex
\documentclass[10pt]{article}
\usepackage{color, amssymb, amsthm, amsmath, amsfonts, ascmac, comment, enumerate}
\usepackage[dvipdfmx]{hyperref}
\usepackage{graphicx}
\usepackage{multirow}
\usepackage{tcolorbox}
\usepackage{xcolor}
\hypersetup{
  colorlinks,
  citecolor=blue!20!black!30!green,
  linkcolor=red,
  urlcolor=blue}
\newtheorem{Thm}{Theorem}
\newtheorem*{Thm*}{Theorem}
\newtheorem{Prop}{Proposition}	
\newtheorem{Lem}{Lemma}

\newtheorem*{Thm1}{\rm\bf Theorem~\ref{Thm_sym_query}}
\newtheorem*{Thm2}{\rm\bf Theorem~\ref{Thm_tight_sym_1}}
\newtheorem*{Thm3}{\rm\bf Theorem~\ref{Thm_tight_sym_2}}
\newtheorem*{Prop1}{\rm\bf Proposition~\ref{Prop_lower_private}}

\newcommand{\QCC}{\mathrm{QCC}}
\newcommand{\CCP}{\mathrm{CC}}
\newcommand{\CCS}{\mathrm{CC}^\textsf{pub}}
\newcommand{\QCCEN}{\mathrm{QCC}^\ast}
\newcommand{\QCCEX}{\QCC_\mathrm{E}}

\newcommand{\Bset}{\{0, 1\}}
\usepackage[top=30truemm,bottom=30truemm,left=20truemm,right=20truemm]{geometry}

\title{Matching upper bounds on symmetric predicates in quantum communication complexity }
\author{Daiki Suruga \thanks{Graduate School of Mathematics, Nagoya University}}
\begin{document}

\maketitle
\input{Abstract.tex}
\input{Introduction.tex}
\input{Preliminaries.tex}
\input{Finding_elements.tex}
\input{Query_comm.tex}
\input{Matching_bounds.tex}

\section*{Acknowledgement}
The author was partially supported by the MEXT Q-LEAP grant No. JPMXS0120319794. The author would like to take this opportunity to thank the “Nagoya University Interdisciplinary Frontier Fellowship” supported by Nagoya University and JST, the establishment of university fellowships towards the creation of science technology innovation, Grant Number JPMJFS2120. The author also would like to thank Fran\c{c}ois Le Gall for his kindness and valuable comments and Ronald de Wolf for kind comments on an earlier draft of this paper.

\bibliography{Citations/quant_info_D, Citations/comm_comp_D, Citations/books_D}
\bibliographystyle{unsrt}
\appendix
\input{Appendix.tex}
\end{document}

%% file: Abstract.tex
\begin{abstract}
In this paper, we focus on the quantum communication complexity of functions of the form $f \circ G = f(G(X_1, Y_1), \ldots, G(X_n, Y_n))$
where $f: \Bset^n \to \Bset$ is a symmetric function, $G: \Bset^j \times \Bset^k \to \Bset$ is any function 
and Alice (resp. Bob)  is given $(X_i)_{i \leq n}$ (resp. $(Y_i)_{i \leq n}$).
Recently, Chakraborty et al. [STACS 2022] showed that the quantum communication complexity of $f \circ G$ is $O(Q(f)\QCCEX(G))$
when the parties are allowed to use shared entanglement, 
where $Q(f)$ is the query complexity of $f$ and $\QCCEX(G)$ is the exact communication complexity of $G$.
In this paper, we first show that the same statement holds \emph{without shared entanglement}, which generalizes their result.
Based on the improved result, we next show tight upper bounds on $f \circ \mathrm{AND}_2$
for any symmetric function $f$ (where $\textrm{AND}_2 : \Bset \times \Bset \to \Bset$ denotes the 2-bit AND function)
in both models: with shared entanglement and without shared entanglement.
This matches the well-known lower bound by Razborov~[Izv. Math. 67(1) 145, 2003] when shared entanglement is allowed
and improves Razborov's bound when shared entanglement is not allowed.
\end{abstract}

%% file: Introduction.tex
\section{Introduction}
\subsection{Motivation}
\paragraph{Communication complexity}
The model of (classical) communication complexity was originally introduced by Yao~\cite{Yao79}.
In this model, there are two players, Alice who receives $x \in \mathcal{X}$ and Bob who receives $y \in \mathcal{Y}$.
Their goal is to compute a known function $f: \mathcal{X} \times \mathcal{Y} \to \Bset$ with as little communication as possible.
Due to this simple structure, lower and upper bounds on communication complexity problems have applications on many other fields such as 
VLSI design, circuit complexity, data structure, etc. (See~\cite{KN96, RY20} for good references.)
Communication complexity has been investigated in many prior works
since its introduction.
\par
In communication complexity, Set-Disjointness ($\textsf{DISJ}_n (x, y) = \neg \bigvee_{i \leq n} (x_i \wedge y_i)$),
Equality ($\textsf{EQ}_n(x, y) = \neg \bigwedge_{i \leq n}(x_i \oplus y_i)$), and Inner-Product function 
($\textsf{IP}_n(x, y) = \bigoplus_{i \leq n}(x_i \wedge y_i)$)
are three of the most well-studied functions.
Denoting the private randomized communication complexity of a function $f$ (with error $ \leq 1/3$) as $\CCP(f)$,
it has been shown that $\CCP(\textsf{DISJ}_n) = \CCP(\textsf{IP}_n) = \Theta(n)$ and $\CCP(\textsf{EQ}_n) = \Theta(\log n)$ hold.
Note that if shared randomness between the two parties is allowed, 
$\CCS(\textsf{DISJ}_n) = \CCS(\textsf{IP}_n) = \Theta(n)$ and $\CCS(\textsf{EQ}_n) = \Theta(1)$ hold
where $\CCS(f)$ denotes the randomized communication complexity of a function $f$ with error $\leq 1/3$ and with shared randomness.
Observing from $\CCP(\textsf{EQ}_n) \neq \CCS(\textsf{EQ}_n)$, 
we see that the shared randomness sometimes enables to reduce the communication complexity.
Therefore, we need to carefully treat the effect of the shared randomness 
when analyzing the communication complexity of functions.
(Note that if $\CCS(f)$ is strictly larger than $O(\log n)$, Newman's theorem~\cite{New91} tells us that $\CCS(f) = O(\CCP(f))$ holds.)
\par
In 1993, Yao~\cite{Yao93} introduced the model of \emph{quantum} communication complexity
 based on the model of classical communication complexity.
The main difference between the classical and quantum model is that Alice and Bob use quantum bits to transmit their information in the quantum model.
As quantum information science has been growing up rapidly, quantum communication complexity
has been widely studied~\cite{BC97, BCvD01, Bra03, BCMdW10}.
In the case of quantum communication complexity,
the three functions mentioned above satisfy
$\QCC(\textsf{DISJ}_n) = \Theta(\sqrt{n})$~\cite{Raz03, AA05}, $\QCC(\textsf{IP}_n) = \Theta(n)$~\cite{CvDNT98} and $\QCC(\textsf{EQ}_n) = \Theta(\log n)$~\cite{BdW01} 
, where $\QCC(f)$ denotes the private quantum communication complexity of a function $f$.
If Alice and Bob have shared entanglement,
$\QCCEN(\textsf{DISJ}_n) = \Theta(\sqrt{n})$~\cite{Raz03, AA05}, $\QCCEN(\textsf{IP}_n) = \Theta(n)$~\cite{CvDNT98} and $\QCCEN(\textsf{EQ}_n) = \Theta(1)$~\cite{BdW01} hold
where $\QCCEN(f)$ denotes the quantum communication complexity of the function $f$ when shared entanglement is allowed.
Even though the power of entanglement is not significant in these examples, 
careful treatment of shared entanglement is important since many non-trivial properties of entanglement have been witnessed (e.g., \cite{CHSH69, Mer90, GKRdW06, Gav08, BCMdW10}), 
including Ref.~\cite{Gav08} that shows Newman's theorem~\cite{New91} does not hold in case of shared entanglement.
\par
\paragraph{Composed functions}
In both classical and quantum communication complexity, many important functions have the form 
\[f \circ G: (X, Y) \mapsto f((G(X_1, Y_1)), \ldots, G(X_n, Y_n)) \in \Bset\]
where $X= (X_i)_{i \leq n} \in \Bset^{nj}$, $Y= (Y_i)_{i \leq n} \in \Bset^{nk}$, $f:\Bset^n \to \Bset$ and $G: \Bset^j \times \Bset^k \to \Bset$.
This fact is already observed in the three of the most well-studied functions: Set-Disjointness ($\neg \textsf{OR}_n \circ \textsf{AND}_2$),
Equality ($\textsf{AND}_n \circ \textsf{XOR}_2$), and Inner-Product function ($\textsf{XOR}_n \circ \textsf{AND}_2$).
As a natural consequence of its importance, functions of this form have been investigated deeply~\cite{SZ09a, LZ10, ACFN12} in both classical and quantum communication complexity.
Even though the functions $f \circ G$ are in general difficult to analyze in detail because of their generality, 
the analysis may become simpler when $G$ has a simpler form.
Let us explain in detail about upper and lower bounds on the quantum communication complexity when $G$ is a simple function such as $\textsf{AND}_2$, $\textsf{XOR}_2$.
In the case of upper bounds,  Buhrman et al.~\cite{BCW98} showed $\QCC(f \circ G) = O(\textrm{Q}(f) \log n)$ holds when $G \in \{\textsf{AND}_2$, $\textsf{XOR}_2\}$,
where $\textrm{Q}(f)$ denotes the bounded error query complexity of a function $f$.
Applying this result, we immediately get $\QCC(\textsf{DISJ}_n) = O(\sqrt{n} \log n)$ because $Q(\textsf{OR}_n) = O(\sqrt{n})$ holds by Grover's algorithm.
This is an important result since it shows that the fundamental function $\textsf{DISJ}_n$ can be computed more efficiently than in classical scenario (recall $\CCS(\textsf{DISJ}_n) = \Theta(n)$).
This upper bound $\QCC(\textsf{DISJ}_n) = O(\sqrt{n} \log n)$ was later improved by~\cite{HdW02} and finally improved to $O(\sqrt{n})$ by~\cite{AA05}.
Ref.~\cite{BCW98} gives many important upper bounds for functions $f \circ G$.
On the other hand, Razborov~\cite{Raz03} treated lower bounds of $\QCCEN(f \circ G)$ and showed several tight bounds 
when $f$ is a symmetric function and $G$ is $\textsf{AND}_2$. 
For example, Ref.~\cite{Raz03} shows $\QCCEN(\textsf{DISJ}_n) = \Omega(\sqrt{n})$ and $\QCCEN(\textsf{IP}_n) = \Omega({n})$.
Combining the $O(\sqrt{n})$ bound~\cite{AA05} and $\Omega(\sqrt{n})$ bound~\cite{Raz03} imply $\QCC(\textsf{DISJ}_n) = \Theta(\sqrt{n})$.
Our contributions can be understood as a generalization of these works~\cite{BCW98, Raz03, AA05}.
\par
As described above, the relation $\QCC(f \circ G) = O(\textrm{Q}(f) \log n)$ holds when the function $G$ is either $\textsf{AND}_2$ or $\textsf{XOR}_2$~\cite{BCW98},
and this upper bound was then improved to $O(\sqrt{n})$ by Aaronson and Ambainis~\cite{AA05} when $f = \textsf{OR}_n$.
This implies that the $\log n$ factor in~\cite{BCW98} is not required in the case of Set-Disjointness function.
Considering this fact, one may wonder whether the $\log n$ overhead is not required for arbitrary function when $G \in \{\textsf{AND}_2, \textsf{XOR}_2\}$.
Chakraborty et al.~\cite{CCMP20} treated this problem and gave a negative answer.
They exhibited a function $f$ that requires $\Omega(Q(f)\log n)$ communication to compute $f \circ \textsf{XOR}_2$.
This means that the upper bound $O(Q(f) \log n)$ in~\cite{BCW98} is tight for generic functions.
Interestingly, their subsequent work~\cite{CCH+22} generalized the result and proved the $\log n$ overhead is not required when $f$ is a symmetric function.
In this paper, we focus on functions of the form $\textsf{SYM} \circ G$ where $\textsf{SYM}$ is a symmetric function.
As described below in Section~\ref{subsec_first_result} and Section~\ref{subsec_second_result}, our first result generalizes the paper~\cite{CCH+22} 
and our second result shows a tight lower and upper bound on the quantum communication complexity of such functions $\textsf{SYM} \circ G$ when $G = \textsf{AND}_2$.
\par

\subsection{First result: On improving the result~\cite{CCH+22}}\label{subsec_first_result}
As mentioned above, the paper~\cite{CCH+22} showed that the $\log n$ factor in $O(Q(f) \log n)$ upper bound
is not required when we focus on a symmetric function $f = \textsf{SYM}$.
More precisely, it is shown in Ref.~\cite{CCH+22} that there exists a protocol for a function $\textsf{SYM} \circ G$ with $O(Q(\textsf{SYM})\QCCEX(G))$ qubits of communication ($\QCCEX(G)$ denotes the exact communication complexity of $G$) which uses \emph{shared entanglement}.
Even though the amount of shared entanglement in their protocol is not so large, 
there are cases when the amount of the entanglement is significantly larger than the communication cost $O(Q(\textsf{SYM})\QCCEX(G))$ as stated in~\cite[Remark 4]{CCH+22}.
Thus, in general the shared entanglement  can not be included as a part of the communication in their protocol.
We improve their result and show that the same statement holds even without any shared entanglement.
That is, we show the following theorem.
\begin{Thm}\label{Thm_sym_query}
For any symmetric function $f: \Bset^n \to \Bset$ and any two-party function $G: \Bset^j \times \Bset^k \to \Bset$,
\begin{equation*}
\QCC(f \circ G) \in O(Q(f)\QCCEX(G)).
\end{equation*}
\end{Thm}
\paragraph{Proof technique}
In the paper~\cite{CCH+22}, the desired protocol is constructed by employing a new technique called \emph{noisy amplitude amplification}, 
which needs a certain amount of entanglement shared between Alice and Bob.
Based on the noisy amplitude amplification technique, Ref.~\cite{CCH+22} shows the following theorem.
\begin{Thm*}[{\cite[Theorem~21]{CCH+22}}]
Suppose Alice (resp. Bob)  is given $(X_i)_{i \leq n} \in \{0, 1\}^{jn}$ (resp. $(Y_i)_{i \leq n}\in \{0, 1\}^{kn}$).
There is a protocol which satisfies the followings:
\begin{itemize}
\item The protocol uses $O(\sqrt{n} \QCCEX(G))$ qubits of communication and $\lceil \log n\rceil$ EPR pairs.
\item The protocol finds the coordinate $i$ satisfying $G(X_i, Y_i) = 1$ with probability $99/100$ when such $i$ exists, 
and outputs ``No'' with probability $1$ when no such $i$ exists.
\end{itemize}
\end{Thm*}
Using this protocol as a subroutine, the authors of Ref.~\cite{CCH+22} constructed the main protocol for $f \circ G$, which inherently requires a certain amount of the entanglement.
\par
On the other hand, in the case of Set-Disjointness, Aaronson and Ambainis~\cite{AA05} showed a protocol with $O(\sqrt{n})$ qubits of communication
which \emph{does not use any shared entanglement} but does find a coordinate $i$ satisfying $x_i \wedge y_i = 1$ with probability $99/100$. 
Based on the construction of the protocol in~\cite{AA05} rather than the noisy amplitude amplification technique used in~\cite{CCH+22},
we successfully construct a generalized version of the above theorem in Proposition~\ref{Prop_finding_many} which does not require any shared entanglement.
Once we show the generalized version, the rest is shown in a similar manner as in~\cite{CCH+22}, which is described in Section~\ref{sec_query_comm}.
Thus, we obtain the protocol for $\textsf{SYM} \circ G$ using $O(Q(\textsf{SYM})\QCCEX(G))$ qubits which does not use any shared entanglement.

\subsection{Second result: On tight upper bounds for $\textsf{SYM} \circ \textsf{AND}_2$}\label{subsec_second_result}
In our second result, we focus on tight upper bounds on the quantum communication complexity of $\textsf{SYM} \circ \textsf{AND}_2$.
We first note here that the paper~\cite{CCH+22} and our first result already exhibit protocols with $O(Q(\textsf{SYM}))$ qubits which are more efficient 
than the protocol in~\cite{BCW98} with $O(Q(\textsf{SYM)} \log n)$ qubits.
However, even a protocol with $O(Q(\textsf{SYM}))$ qubits of communication does not generally give a tight upper bound.
For example, the quantum communication complexity of $\textsf{AND}_n \circ \textsf{AND}_2$ is $O(1)$ but $Q(\textsf{AND}_n) = \Theta(\sqrt{n})$.
Therefore, we need to develop another technique to show a tight upper bound.
\par
In this framework, Razborov~\cite{Raz03} and Sherstov~\cite{She11} showed the following strong result.
\begin{Thm*}[{\cite{Raz03, She11}}]
Let $\textsf{SYM}_n : \Bset^n \to \Bset$ be a symmetric function and $D:\{0, \ldots, n\} \to \Bset$ be a function satisfying\footnote{Note that for any symmetric function $f$, there is a corresponding function $D$ satisfying $f(x) = D(|x|)$ where $|x|$ denotes the Hamming weight of a bit string $x$.} $\textsf{SYM}_n(x) = D(|x|)$.
Define
\begin{eqnarray*}
l_0(D) &=& \max\big\{l \:|\:  1 \leq l\leq n/2 \text{~and~}D(l) \neq D(l - 1)\big\},\\
l_1(D) &=& \max\big\{n - l \: |\: n/2 \leq l < n \text{~and~}D(l) \neq D(l + 1)\big\}.
\end{eqnarray*}
Then we have $\QCCEN(\textsf{SYM}_n \circ \textsf{AND}_2) \in \Omega(\sqrt{n l_0(D)} + l_1(D))$
and $\QCC(\textsf{SYM}_n \circ \textsf{AND}_2) \in O(\{\sqrt{n l_0(D)} + l_1(D)\}\log n)$.
\end{Thm*}
This theorem already shows the nearly tight bound $\QCCEN(\textsf{SYM}_n \circ \textsf{AND}_2) = \tilde{\Theta}(\sqrt{n l_0(D)} + l_1(D))$ up to a multiplicative $\log n$ factor.
To show an exact tight upper bound, it is thus sufficient to create a protocol with $O(\sqrt{nl_0(D)} + l_1(D))$ qubits of communication 
by removing the $\log n$ factor.
In this paper, we successfully show that the multiplicative $\log n$ factor is not required in the model with shared entanglement.
That is, we get the following theorem.
\begin{Thm}\label{Thm_tight_sym_1}
For any symmetric function $\textsf{SYM}_n: \Bset^n \to \Bset$,
$\QCCEN(\textsf{SYM}_n \circ \textsf{AND}_2) \in O(\sqrt{nl_0(D)} + l_1(D))$ holds.
\end{Thm}
In the model without shared entanglement, we also show a similar statement, albeit with an additive $\log \log n$ factor.
Thus we show

\begin{Thm}\label{Thm_tight_sym_2}
For any symmetric function $\textsf{SYM}_n: \Bset^n \to \Bset$,
$\QCC(\textsf{SYM}_n \circ \textsf{AND}_2) \in O(\sqrt{nl_0(D)} + l_1(D) + \log \log n)$ holds.
\end{Thm}
This shows, for the first time,  the tight relation $\QCCEN(\textsf{SYM}_n \circ \textsf{AND}_2) = \Theta(\sqrt{n l_0(D)} + l_1(D))$ in the model with shared entanglement, matching the lower bound by~\cite{Raz03, She11}.
In the model without shared entanglement, however, there is still a $\log \log n$ gap between the communication cost of our protocol and the lower bound~\cite{Raz03, She11}.
To fill this gap, we also show that our protocol without shared entanglement is in fact optimal:

\begin{Prop}\label{Prop_lower_private}
For any non-trivial symmetric function $f_n: \Bset^n \to \Bset$,  
\begin{itemize}
\item if the function $f_n$ satisfies $l_0(D_{f_n}) > 0$ or $l_1(D_{f_n}) > 1$,
$\QCC(f_n \circ \textsf{AND}_2) \in \Omega(\sqrt{n l_0(D_{f_n})} + l_1(D_{f_n}) + \log \log n)$ holds.
\item Otherwise (If $f_n$ satisfies $l_0(D_{f_n}) = 0$ and $l_1(D_{f_n}) \leq 1$), $\QCC(f_n \circ \textsf{AND}_2) \in \Theta(1)$ holds.
\end{itemize}
\end{Prop}

In the proof of Proposition~\ref{Prop_lower_private}, the fooling set argument, a standard technique in communication complexity, plays a fundamental role.
\paragraph{Proof technique}
Let us now explain the main idea for the desired protocol used in Theorem~\ref{Thm_tight_sym_1} and Theorem~\ref{Thm_tight_sym_2}.
To create the desired protocol for $\textsf{SYM} \circ \textsf{AND}_2$, we first decompose the symmetric function $\textsf{SYM}(x) = D(|x|)$ into
the two symmetric functions $\textsf{SYM}_0 (x) := D_0(|x|)$ and $\textsf{SYM}_1 (x) := D_1(|x|)$ as follows:
\begin{equation*}
D_0(m) :=
\begin{cases}
D(m) &\text{if $m \leq l_0(D)$}\\
0  &\text{otherwise}
\end{cases}, \quad
D_1(m) =
\begin{cases}
D(m) &\text{if $m > n - l_1(D)$}\\
0  &\text{otherwise}
\end{cases}.
\end{equation*}
Note that the function $D$ takes a constant value on the interval $[l_0(D), n -l_1(D)]$.
As discussed in Section~\ref{sec_upper_sym}, it turns out that computing $\textsf{SYM}_0 \circ \textsf{AND}_2$ and $\textsf{SYM}_1 \circ \textsf{AND}_2$
separately is enough to compute the entire function $\textsf{SYM} \circ \textsf{AND}_2$.
Therefore, we only need to design two distinct protocols: one protocol for $\textsf{SYM}_0 \circ \textsf{AND}_2$ and the other protocol for $\textsf{SYM}_1 \circ \textsf{AND}_2$.
We now explain how to design the two protocols.
\begin{itemize}
\item To compute $\textsf{SYM}_0 \circ \textsf{AND}_2$, we simply use our first result. This uses $O(\sqrt{nl_0(D)})$ qubits of communication
since $Q(\textsf{SYM}_0) = O(\sqrt{n l_0(D)})$ holds~\cite{Pat92, BBC+01}.
\item To compute $\textsf{SYM}_1 \circ \textsf{AND}_2$, Alice and Bob directly compute the number of elements in the set $\{i \leq n \mid \textsf{AND}_2(x_i, y_i) = 1\}$
under the condition\footnote{If the condition does not hold, $\textsf{SYM}_1 \circ \textsf{AND}_2(x, y)$ must be zero. Alice and Bob check this condition with only two bits of communication. } $\min\{|x|, |y|\} \geq n -l_0(D)$. 
By taking the negation on the inputs, this problem is reduced to the computation of the number of elements in the set $\{i \leq n \mid x_i = 0 \text{~or~} y_i = 0\}$
under the condition $\min\{|x|, |y|\} \leq l_0(D)$. 
In fact, this problem and related problems have been analyzed in several works~\cite{HSZZ06, BCK+14, BCKWY16, HPZZ20}
and it is shown in~\cite{BCK+14} that $O(l_0(D))$ classical communication is sufficient when shared randomness is allowed
(and the additional $O(\log \log n)$ bits of communication\footnote{
In this case, $\min\{|x|, |y|\} \geq n -l_0(D)$ holds and therefore Newman's theorem tells us that $O(\log \log \#\{x \mid |x|\geq n - l_0(D)\})$ bits simulates the shared randomness. 
As shown in Section~\ref{sec_upper_sym}, the additional bits required are in fact bounded by $O(\log \log n)$.} 
are required to convert the shared randomness into private randomness).
\end{itemize}
Combining the above protocols, we create the desired protocol for $\textsf{SYM} \circ \textsf{AND}_2$ with $O(\sqrt{nl_0(D_f)} + l_1(D_f))$ communication.
One thing which should be noted is that as seen in the above protocol, 
what Alice and Bob needed to share beforehand is shared randomness, not shared entanglement.
This means that we in fact show the upper bound $O(\sqrt{nl_0(D_f)} + l_1(D_f))$ in a weaker communication model 
where shared randomness is allowed but shared entanglement is not allowed.
\par

\subsection{Organization of the paper}
In Section~2, we list several notations and facts used in this paper.
In Section~3, we generalize the protocol for Set-Disjointness~\cite{AA05} and create a useful protocol which is used for our main results.
In Section~4, we treat the first result and show Theorem~\ref{Thm_sym_query}.
In Section~5, we treat the second result and show Theorem~\ref{Thm_tight_sym_1} and Theorem~\ref{Thm_tight_sym_2}.

%% file: Preliminaries.tex
\section{Preliminaries}
For any function $f$, we denote the quantum communication complexity of zero-error protocols, 
the bounded-error  quantum communication complexity (with error $\leq 1/3$) \emph{without shared entanglement},
the bounded-error  quantum communication complexity (with error $\leq 1/3$) with \emph{shared entanglement} of a function $f$ 
by $\QCCEX(f), \QCC(f)$ and $\QCCEN(f)$ respectively.
Trivially, it holds that
$\QCCEN(f) \leq \QCC(f) \leq \QCCEX(f).$
We also denote the bounded-error query complexity of a function $f$ by $\mathrm{Q}(f)$.
For a $n$-bit string $x$, we denote the bitwise negation of $x$ by $\neg x = (\neg x_1, \ldots, \neg x_n)$.

\paragraph{Symmetric function}
Here we list several important facts about symmetric functions.
For any symmetric function $f$, $f$ can be represented as $f(x) = D_f(|x|)$ using some function $D_f:\{0, 1, \ldots, n\} \to \Bset$.
Denoting
\begin{eqnarray*}
l_0(D_f) &=& \max\big\{l \:|\:  1 \leq l\leq n/2 \text{~and~}D_f(l) \neq D_f(l - 1)\big\},\\
l_1(D_f) &=& \max\big\{n - l \: |\: n/2 \leq l < n \text{~and~}D_f(l) \neq D_f(l + 1)\big\},
\end{eqnarray*}
prior works~\cite{Pat92, BBC+01} show that the query complexity $Q(f)$ of a symmetric function $f$ is characterized as $Q(f) = \Theta(\sqrt{n(l_0(D_f) + l_1(D_f))})$.

%% file: Finding_elements.tex
\section{Communication cost for finding elements}\label{sec_find_elements}
This section is devoted to show Proposition~\ref{Prop_finding_many}, which is the quantum communication version of \cite[Theorem~5.16]{AA05}.
\begin{Prop}\label{Prop_finding_many} 
There is a protocol $\textsf{FIND-MORE}_k$ using $O(\sqrt{\frac{n}{k}} \QCCEX(G))$ qubits and using shared randomness which satisfies the following:
\begin{itemize}
\item The protocol outputs a coordinate $i \in [n]$ such that $G(X_i, Y_i) = 1$ w.p. $\geq 99/100$ when there exist at least $k$ such coordinates.
\item The protocol answers ``there is no such coordinate" w.p. $1$ when there is no such coordinate.
\item The protocol does not use any shared entanglement.
\end{itemize}

\end{Prop}
The proof is given in Section~\ref{subsec_proof_finding_many}.
\subsection{A key lemma}
To show Proposition~\ref{Prop_finding_many}, we first show the following lemma:
\begin{Lem}\label{Lem_finding_exact}
For $\gamma \in \mathbb{N}$, there is a protocol $\textsf{FIND-EXACT}_\gamma$ 
using $O(\sqrt{\frac{n}{\gamma}} \QCCEX(G))$ qubits and shared randomness which satisfies the followings:
\begin{itemize}
\item The protocol outputs a coordinate $i \in [n]$ such that $G(X_i, Y_i) = 1$ w.p. $\geq 99/100$ when there exist exactly $k$ such coordinates for some $k$ satisfying $3k/2 < \gamma < 3k$.
\item The protocol answers ``there is no such coordinate" w.p. $1$ when there is no such coordinate.
\item The protocol does not use any shared entanglement.
\end{itemize}
\end{Lem}
In the proof of Lemma~\ref{Lem_finding_exact}, we use Lemma~\ref{Lem_finding_one_element} which is a modified protocol of the one given in \cite[Section~7]{AA05}.
See Appendix~\ref{App_modify} for the modification.
\begin{Lem}\label{Lem_finding_one_element}
There is a protocol $\textsf{FIND-ONE}$ with $O(\sqrt{n} \QCCEX(G))$ cost which satisfies the followings:
\begin{itemize}
\item The protocol outputs the coordinate $i \in [n]$ such that $G(X_i, Y_i) = 1$ w.p. $\geq 99/100$ when such $i$ exists.
\item The protocol answers ``there is no such coordinate" w.p. $1$ when there is no such coordinate.
\item The protocol does not use any shared entanglement.
\end{itemize}
\end{Lem}

\begin{proof}[Proof of Lemma~\ref{Lem_finding_exact}]
We first divide the set $\{1, \ldots, n\}$ into $n/\gamma$ subsets $A_j = \{(j - 1)\gamma + 1, \ldots, j\gamma\}~(1 \leq j \leq n/\gamma)$, each containing $\gamma$ sub-inputs.
Using shared randomness, Alice and Bob pick the set of coordinates $\{i_1, \ldots, i_{n/\gamma}\} \subset [n]$
where each $i_j$ is chosen uniformly at random from the set $A_j$.
Alice and Bob then perform the protocol $\textsf{FIND-ONE}$ pretending the inputs are $(X_{i_1}, \ldots, X_{i_{n/\gamma}})$
for Alice and $(Y_{i_1}, \ldots, Y_{i_{n/\gamma}})$ for Bob.
Since $\textsf{FIND-ONE}$ requires $O(\sqrt{n}\QCCEX(G))$ qubits of communication for the input length $n$,
this protocol with the input length $n/\gamma$ requires $O(\sqrt{\frac{n}{\gamma}}\QCCEX(G))$ qubits of communication.
\par
We now analyze the correct probability of this protocol, following the technique used in \cite[Lemma~5.15]{AA05}.
Assume there exist exactly $k$ coordinates satisfying $G(X_i, Y_i) = 1$ and $3k/2 < \gamma < 3k$ holds. 
Suppose $i_0$ satisfies $G(X_{i_0}, Y_{i_0}) = 1$. Then the coordinate $i_0$ is chosen as the shared randomness w.p. $1/\gamma$.
Given that $i_0$ is chosen, one of other coordinates $i'$ satisfying $G(X_{i'}, Y_{i'}) = 1$ is chosen w.p. 0 if $i_0, i'$ are in the same subset $A_j (1 \leq j \leq n/\gamma)$
and w.p. $1/\gamma$ if $i_0$ and $i'$ are in two different subsets. Therefore, the probability of ``the coordinate $i_0$ alone is chosen'' is at least
\begin{equation*}
\frac{1}{\gamma}\left(1 - \frac{k -1}{\gamma}\right)
\geq 
\frac{1}{\gamma}\left(1 - \frac{k}{\gamma}\right).
\end{equation*}
Considering the events ``the coordinate $i_0$ is chosen" are mutually disjoint, we see that the probability of ``exactly one such coordinate is chosen'' is at least
$k/\gamma - (k/\gamma)^2$. Since $3k/2 < \gamma < 3k$ holds, we observe that the probability is at least $2/9$.
This shows the event ``at least one element is chosen'' occurs w.p. $\geq 2/9$.
\par
Therefore, by the property of $\textsf{FIND-ONE}$, our new protocol satisfies the followings:
\begin{itemize}
\item The protocol outputs the coordinate $i \in [n]$ such that $G(X_i, Y_i) = 1$ w.p. $\Omega(1)$ when there exist exactly $k$ such coordinates for some $k$ satisfying $3k/2 < \gamma < 3k$.
\item The protocol answers ``there is no such coordinate" w.p. $1$ when there is no such coordinate.
\item The protocol does not use any shared entanglement.
\end{itemize}
To amplify the success probability $\Omega(1)$ to $99/100$, Alice and Bob perform this above protocol recursively 
while at each repetition checking if the output $i_\mathrm{out}$ satisfies $G(X_{i_\mathrm{out}}, Y_{i_\mathrm{out}}) = 1$. 
This repetition uses only some constant overhead on the communication cost and hence we obtain the desired statement.
\end{proof}

\subsection{Proof of Proposition~\ref{Prop_finding_many}}\label{subsec_proof_finding_many}
Using the protocol $\textsf{FIND-EXACT}_\gamma$, we show Proposition~\ref{Prop_finding_many} as follows.
\begin{proof}[Proof of Proposition~\ref{Prop_finding_many}]
The protocol $\textsf{FIND-MORE}_k$ is executed as follows:
\begin{enumerate}[(1)]
\item For $j = 0$ to $\log_2 (n/k)$,  Alice and Bob perform $\textsf{FIND-EXACT}_{\gamma_j}$ where $\gamma_j = 2^j k$.
\item As shared randomness, Alice and Bob pick one coordinate $i$ uniformly at random from the set $[n]$
and check if $G(X_i, Y_i) = 1$. This is repeated for $O(1)$ times.
\end{enumerate}
We first analyze the communication cost of this protocol. The first step requires
\begin{equation*}
\sum_{j = 0}^{\log_2 (n/k)} O\left(\sqrt{\frac{n}{2^j k}} \QCCEX(G)\right)
= O\left(\sqrt{\frac{n}{k}}\QCCEX(G)\right) \sum_{j = 0}^{\log_2 (n/k)} \frac{1}{2^{j/2}}
= O\left(\sqrt{\frac{n}{k}}\QCCEX(G)\right)
\end{equation*}
qubits of communication.
The second step requires $O(\QCCEX(G))$ qubits of communication.
Therefore, in total, $O\left(\sqrt{\frac{n}{k}}\QCCEX(G)\right)$ qubits are used in this protocol.
\par
Next we analyze the correct probability of this protocol. Let $k^\ast \geq k$ be the number of coordinates satisfying $G(X_i, Y_i) = 1$.
If $k^\ast \leq n/3$, then there exists $j$ satisfying $3k^\ast/2 < \gamma_j < 3k^\ast$. 
Therefore, $\textsf{FIND-EXACT}_{\gamma_j}$ finds the desired coordinate w.p. $\geq 99/100$.
On the other hand, if $k^\ast > n/3$, the second step finds the desired coordinate w.p. $1/3$.
Then O(1) repetitions increase the success probability to $99/100$.
\end{proof}

%% file: Query_comm.tex
\section{Communication protocol for symmetric functions}\label{sec_query_comm}

In \cite[Theorem~22 and Theorem~25]{CCH+22}, the following theorem has been shown (with a slightly different expression):
\begin{Thm*}[{\cite[Theorem~22 and Theorem~25]{CCH+22}}]
Suppose $\textsf{FIND-MORE}_k$ uses $m$ EPR-pairs as shared entanglement and \emph{arbitrarily much} shared randomness.
Then for any symmetric function $f: \Bset^n \to \Bset$ and any two-party function $G: \Bset^j \times \Bset^k \to \Bset$,
there is a protocol with $O(Q(f)\QCCEX(G))$ qubits which satisfies the followings:
\begin{itemize}
\item The protocol successfully computes $f \circ G$ with probability $\geq 99/100$.
\item The protocol uses $m \cdot O(l_0(D_f) + l_1(D_f))$ EPR-pairs as shared entanglement.
\item The protocol uses $O(\log n)$ bits of shared randomness.
\end{itemize}
\end{Thm*}

As is shown in Proposition~\ref{Prop_finding_many}, our modified protocol $\textsf{FIND-MORE}_k$ does not use any shared entanglement.
Therefore,  we set $m = 0$ in the statement above and obtain the following theorem.
(Note that $O(\log n)$ bits of shared randomness are included in a part of communication since the $O(\log n)$ bits
are negligible compared to $Q(f) \geq \Omega(\sqrt{n})$ when $f$ is not trivial.)

\begin{Thm1}
For any symmetric function $f: \Bset^n \to \Bset$ and any two-party function $G: \Bset^j \times \Bset^k \to \Bset$,
\begin{equation*}
\QCC(f \circ G) \in O(Q(f)\QCCEX(G)).
\end{equation*}
\end{Thm1}

%% file: Matching_bounds.tex
\section{Tight upper bound for symmetric functions}\label{sec_upper_sym}
In this section, we show the following two theorems:

\begin{Thm2}
For any symmetric function $\textsf{SYM}_n: \Bset^n \to \Bset$,
$\QCCEN(\textsf{SYM}_n \circ \textsf{AND}_2) \in O(\sqrt{nl_0(D)} + l_1(D)$ holds.
\end{Thm2}
\begin{Thm3}
For any symmetric function $\textsf{SYM}_n: \Bset^n \to \Bset$,
$\QCC(\textsf{SYM}_n \circ \textsf{AND}_2) \in O(\sqrt{nl_0(D)} + l_1(D) + \log \log n)$ holds.
\end{Thm3}
To show these theorems, we use the following protocol that is a modification of the protocol given in \cite[Theorem~3.1]{BCK+14}.
For completeness, we describe the modification in Appendix~\ref{App_modify_set}.
\begin{Prop}\label{Prop_intersection}
Suppose the inputs $x, y \in \Bset^n$ satisfy $\max\{|x|, |y|\} \leq k$.
There is a public coin classical protocol\footnote{
Note that this protocol may use many amount of shared randomness.
} with $O(k)$ bits of communication 
which computes the set $\{i |x_i = y_i = 1\} \subset [n]$ w.p. $99/100$.
\end{Prop}
Following the technique used in~\cite[Section~4]{Raz03}, we prove Theorem~\ref{Thm_tight_sym_1} and Theorem~\ref{Thm_tight_sym_2} as follows:
\begin{proof}[Proof of Theorem~\ref{Thm_tight_sym_1} and Theorem~\ref{Thm_tight_sym_2}]
Let us first describe some important facts based on the arguments in~\cite{Raz03,She11}.
For any symmetric function $f_n$, the corresponding function $D_{f_n}$ is
constant on the interval $[l_0(D_{f_{n}}), n - l_1(D_{f_{n}})]$.
Without loss of generality, assume $D_{f_{n}}$ takes $0$ on the interval.
(If $D_{f_{n}}$ takes $1$ on the interval, we take the negation of $D_{f_{n}}$.)
Defining $D_0$ and $D_1:\{0, \ldots, n\} \to \Bset$ as
\begin{equation*}
D_0(m) =
\begin{cases}
D_{f_{n}}(m) &\text{if $m \leq l_0(D_{f_{n}})$}\\
0  &\text{otherwise}
\end{cases},
D_1(m) =
\begin{cases}
D_{f_{n}}(m) &\text{if $m > n - l_1(D_{f_{n}})$}\\
0  &\text{otherwise}
\end{cases},
\end{equation*}
$D_{f_{n}} = D_0 \vee D_1$ holds.
Therefore, by defining $f^0_{n}(x) := D_0(|x|)$ and
$f^1_{n}(x) := D_1(|x|)$, we get $f_{n} \circ \textsf{AND}_2 = (f^0_{n} \circ \textsf{AND}_2) \vee (f^1_{n} \circ \textsf{AND}_2)$.
This means, computing $f^{0}_{n} \circ \textsf{AND}_2$ and $f^{1}_{n} \circ \textsf{AND}_2$ separately is sufficient to compute the entire function $f_{n} \circ \mathsf{AND}_2$.
As another important fact needed for our explanation, we note that the query complexity of $f^{0}_{n}$ equals to $O(\sqrt{n l_0(D_{f_{n}})})$
which is proven in \cite{Pat92}.
\par
From now on, we describe two protocols: one protocol for the computation of $f^0_n$ and the other one for the computation of $f^1_n$.
\begin{itemize}
\item {\bf Protocol for $f^0_n$}: We simply apply the protocol of Theorem~\ref{Thm_sym_query} with $G = \mathsf{AND}_2$ (note that $f^0_n$ is a symmetric function).
This protocol uses $O(\sqrt{n l_0(D_{f_{n}})})$ qubits because $Q(f_n^1) = \Theta(\sqrt{n l_0(D_{f_{n}})})$ holds.
\item {\bf Protocol for $f^1_n$}:  First, Bob sends Alice one bit: $1$ if $|\neg y| \leq l_1(D_{f_n})$ and $0$ otherwise.
If Alice receives $1$ and $|\neg x| \leq l_1(D_{f_n})$ holds, they perform the protocol of Proposition~\ref{Prop_intersection} with the inputs $\neg x$ and $\neg y$.
Otherwise, $\min\{|x|, |y|\} < n - l_0(D_{f_n})$ holds and therefore $f^0_n \circ \mathsf{AND}_2(x, y)$ must be zero by the definition of $D_1$.
After the execution of the protocol of Proposition~\ref{Prop_intersection}, Alice and Bob know the set $\{i \leq n \mid x_i = y_i = 0\}$.
Next, Alice sends $|\neg x|$ and Bob sends $|\neg y|$ using $\log l_0(D_{f_n})$ communication,  
and they finally compute $\#\{i\leq n \mid x_i = y_i = 1\}$ as $\#\{i\leq n \mid x_i = y_i = 1\} = n  + \#\{i \leq n \mid x_i = y_i = 0\} - |\neg x| - |\neg y|$.
This protocol uses $O(l_1(D_{f_n}))$ communication bits. 
\par
We then evaluate the cost for public coins.
Even though the execution of this protocol may require much shared randomness, Newman's theorem~\cite{New91} ensures that $O(\log\log |S|)$ bits are
sufficient when the inputs $x, y$ belong to a set $S$. Since $|\neg x|, |\neg y| \leq l_1(D_{f_n})$ holds when executed and using the fact
$\#\{x \in \Bset^n \mid |\neg x| \leq k\} \leq n^k$, 
we conclude that $O(\log (\log n^{l_1(D_{f_n})})) = O(\log l_1(D_{f_n}) + \log \log n)$ bits of shared randomness are sufficient.
Moreover, since $O(\log l_1(D_{f_n}))$ bits of shared randomness are negligible compared to $O(l_1(D_{f_n}))$ bits in communication 
and therefore included as a part of communication with no additional communication cost, we only need to use $O(\log \log n)$ bits as a shared randomness.

\end{itemize}
Combining these two protocols, we get the desired protocol with $O(\sqrt{n l_0(D_{f_n})} + l_1(D_{f_n}))$ cost which uses $ O(\log \log n)$ public coins.
This shows $\QCCEN(f_n \circ \textsf{AND}_2) \in O(\sqrt{n l_0(D_{f_n})} + l_1(D_{f_n}))$
and $\QCC(f_n \circ \textsf{AND}_2) \in O(\sqrt{n l_0(D_{f_n})} + l_1(D_{f_n}) + \log \log n)$ 
by Alice sending $O(\log \log n)$ random bits instead of the shared randomness.
\end{proof}

By combining the arguments we showed so far, we obtain the tight bound $\QCCEN(f_n \circ \textsf{AND}_2) \in \Theta(\sqrt{n l_0(D_{f_n})} + l_1(D_{f_n}))$
on the communication model with shared entanglement.
On the model without shared entanglement, our bound $\QCC(f_n \circ \textsf{AND}_2) \in O(\sqrt{n l_0(D_{f_n})} + l_1(D_{f_n}) + \log \log n)$
still have the additive $\log \log n$ difference from the lower bound.
We next show this upper bound is indeed optimal by using a standard technique, the \emph{fooling set} argument.

\begin{Prop1}
For any non-trivial symmetric function $f_n: \Bset^n \to \Bset$,  
\begin{itemize}
\item if the function $f_n$ satisfies $l_0(D_{f_n}) > 0$ or $l_1(D_{f_n}) > 1$,
$\QCC(f_n \circ \textsf{AND}_2) \in \Omega(\sqrt{n l_0(D_{f_n})} + l_1(D_{f_n}) + \log \log n)$ holds.
\item Otherwise (i.e., if $f_n$ satisfies $l_0(D_{f_n}) = 0$ and $l_1(D_{f_n}) \leq 1$), $\QCC(f_n \circ \textsf{AND}_2) \in \Theta(1)$ holds.
\end{itemize}
\end{Prop1}
\begin{proof}
Let us first prove that $\QCC(f_n \circ \textsf{AND}_2) \in \Theta(1)$ holds when $l_0(D_{f_n}) = 0$ and $l_1(D_{f_n}) \leq 1$ hold.
In this case, there are only two types of the functions: $f_n = \textsf{AND}_n$ or $f_n = \neg\textsf{AND}_n$.
In either case of the functions, Alice and Bob only need to send one single bit expressing whether $x = (1, \ldots, 1)$ for Alice ($y = (1, \ldots, 1)$ for Bob).
Therefore we obtain $\QCC(f_n \circ \textsf{AND}_2) \in \Theta(1)$ since a lower bound $\QCC(f_n \circ \textsf{AND}_2) \in \Omega(1)$ is trivial.
\par
The rest is to show 
$\QCC(f_n \circ \textsf{AND}_2) \in \Omega(\sqrt{n l_0(D_{f_n})} + l_1(D_{f_n}) + \log \log n)$ holds assuming $l_0(D_{f_n}) > 0$ or $l_1(D_{f_n}) > 1$.
First, we note that the $\log \log n$ factor becomes negligible comparing to $\sqrt{n l_0(D_{f_n})} + l_1(D_{f_n})$ when $l_0(D_{f_n}) > 0$ holds.
This means that the well-known lower bound $\Omega(\sqrt{n l_0(D_{f_n})} + l_1(D_{f_n}))$~\cite{Raz03} already gives a tight lower bound.
Therefore, we only need to show
$\QCC(f_n \circ \textsf{AND}_2) \in \Omega( l_1(D_{f_n}) + \log \log n)$ holds assuming $l_0(D_{f_n}) = 0$. 
Moreover, the lower bound
$\QCCEN(f_n \circ \textsf{AND}_2) \in \Omega(\sqrt{n l_0(D_{f_n})} + l_1(D_{f_n}))$ shown in~\cite{Raz03} implies
$\QCC(f_n \circ \textsf{AND}_2) \in \Omega(l_1(D_{f_n}))$.
Therefore, it is sufficient to show
$\QCC(f_n \circ \textsf{AND}_2) \in \Omega(\log \log n)$ when $l_0(D_{f_n}) = 0$ and $l_1(D_{f_n}) > 1$ hold.
\par
Assuming $l_0(D_{f_n}) = 0$, $l_1(D_{f_n}) > 1$ and $D_{f_n} \equiv 0$ on $[l_0(D_{f_n}), n - l_1(D_{f_n})]$ without loss of generality,
we show $\QCC(f_n \circ \textsf{AND}_2) \in \Omega(\log \log n)$.
To show this, we use the fooling set argument~\cite{KN96, RY20}.
Define
\begin{equation*}
\mathrm{FS}_n := \{(x, y) \in \{0, 1\}^{n} \times \{0, 1\}^n \mid x = y \text{~and~} |\neg x| = l_1(D_{f_n}) - 1\}.
\end{equation*}
Then we see that for any $(x, y) \in \mathrm{FS}_n$, $f_n \circ \textsf{AND}_2 (x, y) = 1$
and for any $(x, y), (x', y') \in \mathrm{FS}_n$, $(x, y) \neq (x', y')$ implies 
$f_n \circ \textsf{AND}_2 (x, y') =f_n \circ \textsf{AND}_2 (x', y) = 0$.
Therefore, the deterministic communication complexity $\mathrm{DCC}(f_n \circ \textsf{AND}_2)$
satisfies
\begin{equation*}
 \mathrm{DCC}(f_n \circ \textsf{AND}_2) \geq \log_2 |\mathrm{FS}_n|
\end{equation*}
by the fooling set argument.
As shown in~\cite{Kre95}, it is well-known that
$\QCC(f) \geq \log \mathrm{DCC}(f)$ for any function $f$.
Therefore, by observing $|\mathrm{FS}_n| = {n \choose l_1(D_{f_n}) - 1} \geq \Omega(n)$ for $l_1(D_{f_n}) > 1$,
we obtain the desired statement $\QCC(f_n \circ \textsf{AND}_2) \geq \Omega(\log \log n)$.
\end{proof}

%% file: Appendix.tex
\section{Modification for Lemma~\ref{Lem_finding_one_element}}\label{App_modify}
Here we describe how the protocol given in~\cite[Section~7]{AA05} is modified to the protocol in~Theorem~\ref{Lem_finding_one_element}.
In~\cite[Section~7]{AA05}, the authors proposed a protocol that finds $i \in [n]$ such that $x_i \wedge y_i = 1$ 
where Alice is given $x \in \Bset^n$ and Bob is given $y \in\Bset^n$.
In the protocol, Alice and Bob perform the query
\[
O_{\mathrm{AND}} : |i, z\rangle_A |i\rangle_B \mapsto |i, z \oplus (x_i \wedge y_i)\rangle_A |i\rangle_B
\]
for $O(\sqrt{n})$ times and other operations which require $O(\sqrt{n})$ communication.
Since the query operation is implemented using $2$-qubits of communication, 
this protocol requires $2O(\sqrt{n}) + O(\sqrt{n}) = O(\sqrt{n})$ communication.
\par
Our modification for finding $i$ such that $G(X_i, Y_i) = 1$ is simple.
We just replace the query $O_\mathrm{AND}$ to
\[
O_G : |i, z\rangle_A |i\rangle_B \mapsto |i, z \oplus G(X_i, Y_i)\rangle_A |i\rangle_B.
\]
This protocol indeed finds the desired coordinate $i$, which is shown in the same manner as in~\cite[Section~7]{AA05}.
Let us analyze the communication cost of this protocol.
Since $\QCCEX(G)$ denotes the exact communication complexity of $G$, the operation $O_G$ is implemented using $2\QCCEX(G)$ qubits.
(First $\QCCEX(G)$ communication is used to compute $G$ and the second $\QCCEX(G)$ is used to compute reversely and  clear the unwanted registers.)
Other operations are the same as in the original protocol and therefore use $O(\sqrt{n})$ communication.
Considering that the operation $O_G$ is performed for $O(\sqrt{n})$ times, 
we see that our modified protocol uses $O(\sqrt{n}) + \QCCEX(G)O(\sqrt{n}) = O(\QCCEX(G) \sqrt{n})$ qubits of communication.

\section{Modification for Proposition~\ref{Prop_intersection} }\label{App_modify_set}
In \cite[Theorem~3.1]{BCK+14}, the authors originally showed the following.
\begin{Thm}\label{Thm_BCK}
Suppose the inputs $x, y \in \Bset^n$ satisfy $\max\{|x|, |y|\} \leq k$.
There exists an $O(\sqrt{k})$-round constructive randomized classical protocol that outputs the set $\{i \mid x_i = y_i = 1\}$
with success probability $1 - 1/\textrm{poly}(k)$. In the model of shared randomness the total expected communication is $O(k)$.
\end{Thm}
To modify this theorem for Proposition~\ref{Prop_intersection}, we need to take care of the success probability and the \emph{expected} communication.
To take care of the success probability, we first take a sufficiently large constant $k_0$ such that for any $k \geq k_0$, $1/\textrm{poly}(k) \leq 1/200$.
If $k < k_0$ holds,  the parties perform the protocol in Theorem~\ref{Thm_BCK} with the constant $k_0$. This requires $O(k_0)$ expected communication.
Otherwise (i.e., when $k > k_0$ holds), the parties perform the protocol in Theorem~\ref{Thm_BCK} with the constant $k$, which requires $O(k)$ expected communication.
Since $k_0$ is a constant, the protocol by this modification still requires $O(k)$ expected communication with error $\leq 1/200$.
\par
To convert the expected communication to the worst-case communication, we use the Markov's inequality.
Suppose this protocol requires $C \cdot k$ expected communication. Then the probability of ``the communication cost $\geq 200 C \cdot k$'' is 
less than or equal to $1/200$ by the Markov's inequality. We create the desired protocol by Alice and Bob aborting communication when its cost gets $200 C \cdot k$.
This modified protocol still have the success probability $\geq 99/100$, 
since the first modification has the error $1/200$ and the second modification affects the error at most $1/200$.